\documentclass[reqno,centertags,oneside,12pt]{amsart}

\usepackage[T1]{fontenc}
\usepackage{pslatex}
\usepackage{graphics,graphicx,subfigure,color}
\usepackage{xcolor}
\usepackage{bm}

\usepackage{amsmath,amsthm,amscd,amssymb}
\usepackage{mathptmx}
\usepackage{enumerate}
\usepackage{txfonts}
\usepackage{textcomp}

\selectfont

\newcommand{\R}{\bm{R}}

\newcommand{\I}{\bm{I}}
\newcommand{\h}{\bm{H}}
\newcommand{\A}{\mathcal{A}}
\newcommand{\mrm}[1]{\mathrm{#1}}

\newcommand{\co}{\!\colon\thinspace}

\newcommand{\disc}{\sigma_{\textit{disc}}}
\newcommand{\vecr}{\bm{r}}
\newcommand{\vect}{\bm{t}}
\newcommand{\vecp}{\bm{p}}

\theoremstyle{plain}
\newtheorem{thm}{Theorem}
\newtheorem{lem}{Lemma}
\newtheorem{prop}{Proposition}
\newtheorem{cor}{Corollary}

\theoremstyle{remark}
\newtheorem{rem}{Remark}

\makeatletter
\numberwithin{equation}{section}
\makeatother

\begin{document}

\title{On the three-body Schr\"{o}dinger equation with decaying potentials}

\author{Rytis Jur\v{s}\.{e}nas \\ \\
{Institute of Theoretical Physics and Astronomy \\ of Vilnius University, A. Go\v{s}tauto 12, LT-01108}}

\maketitle

\begin{abstract}
The three-body Schr\"{o}dinger operator in the space of square integrable functions is found to be a certain extension of operators 
which generate the exponential unitary group containing a subgroup with nilpotent Lie algebra of length $\kappa+1$, $\kappa=0,1,\ldots$ 
As a result, the solutions to the three-body Schr\"{o}dinger equation with decaying potentials are shown to exist in the commutator 
subalgebras. For the Coulomb three-body system, it turns out that the task is to solve - in these subalgebras - the radial 
Schr\"{o}dinger equation in three dimensions with the inverse power potential of the form $r^{-\kappa-1}$. As an application to Coulombic 
system, analytic solutions for some lower bound states are presented. Under conditions pertinent to the three-unit-charge system, 
obtained solutions, with $\kappa=0$, are reduced to the well-known eigenvalues of bound states at threshold.\medskip{}
\textbf{PACS:} {03.65.Ge, 03.65.Db, 03.65.Fd}
\end{abstract}

\section{\label{sec:Intro}Introduction}

The goal of the present paper is to demonstrate an analytical approach for solving the three-body Schr\"{o}dinger equation
with translation invariant decaying potentials. The three-body Schr\"{o}dinger operator 
$H$ (the Hamiltonian operator, henceforth) is represented by the closure of operator sum $T+V$. The kinetic energy operator $T$ 
is defined so that its closure, denoted by $T$ as well, is a self-adjoint operator on the domain $D(T)$ and acting in 
$L^{2}(\R^{9})$ by

\begin{equation}
T=-\sum_{1\leq i\leq 3}(2m_{i})^{-1}\Delta_{i}.
\label{eq:T}
\end{equation}

\noindent{}The constants (referred to as masses) $m_{i}>0$ ($i=1,2,3$), the Laplacian $\Delta_{i}$ is in three-dimensional vectors 
$\vecr_{i}=(x_{i},y_{i},z_{i})\in\R^{3}$, with absolute value $r_{i}\in[0,\infty)$. The potential energy operator
$V$ is a scalar translation invariant operator of multiplication by $V$, where real function $V$ fulfills several assumptions:
({\bf A1})

\begin{equation}
V=\sum_{1\leq i<j\leq3}V_{ij}(\vecr_{i}-\vecr_{j})\quad \text{and}\quad 
V_{ij}\to0\quad\text{as}\quad\vert \vecr_{i}-\vecr_{j}\vert\to\infty.
\label{eq:V}
\end{equation}

\noindent{}({\bf A2}) $V_{ij}$ is of the differentiability class $C^{\infty}(\R^{3})$ and it is analytic everywhere except, 
possibly, at $\vecr_{i}=\vecr_{j}$ for $i\neq j$. ({\bf A3}) The operator $V$ is assumed to be a symmetric $T$-bounded operator in 
the sense of Kato \cite{Kato51} (see also \cite{Sim00} and the citation therein), with its domain satisfying $D(V)\supset D(T)$. This
assumption ensures the self-adjointness of $H$ on $D(T)$. 

Let $\nabla_{ij}$ be the gradient in vectors $\vecr_{ij}=\vecr_{i}-\vecr_{j}$ with the absolute value $r_{ij}$.
If one defines the sum $\sum_{i<j}\nabla_{ij}$ by $G$ (Lemma~\ref{lem:lem1}), with $G_{z}$ its $z$-component, and $V_{n}$ by
$G_{z}^{n}V$ ($n=0,1,\ldots$), then there exists a subset 
$\I_{\kappa}\subset\R^{6}$ (\S\ref{subsec:Untary}, eq.~(\ref{eq:Inu})) such that operators $G_{z}$, $V_{0}\equiv V$, $V_{1}$, $\ldots$, 
$V_{\kappa}$ form a $(\kappa+2)$-dimensional nilpotent Lie algebra $\A$ in $\I_{\kappa}$ (see Theorem~\ref{thm:main}), whereas 
$V_{\kappa+p}$ is the operator of multiplication by zero for all integers $p=1,2,\ldots$ For a particular Coulomb three-body system, the 
latter leads to a 
well-known observation that bound states exist whenever one of the three charges has a different sign (see, for example, 
\cite{Bha87,Fro92,Mar92}), though this is not a sufficient condition of boundedness (see also \S\ref{sec:stability} or, in particular, 
eq.~(\ref{eq:zeroV})). 

After establishing the nilpotency of $\A$ and thereby the existence of its commutator subalgebras $\A^{c}\subseteq\A$ 
($c=0,1,\ldots,\kappa+1$) we come to an important conclusion that the eigenvalue $E$ of $H$ in $L^{2}(\R^{9})$ is the sum of 
$E_{\kappa}$, the eigenvalues of $H_{\kappa}$ in $L^{2}(\I_{\kappa})$, where $\kappa$ runs, in general, from $0$ to $\infty$ 
(eq.~(\ref{eq:E})), whereas the eigenvalue equation for $H_{\kappa}$ can be decomposed, to some extent, into two equations in 
distinct subspaces so that their common solutions with respect to the eigenvalues were equal to $E_{\kappa}^{0}$, eq.~(\ref{eq:spec0}).
Then the eigenvalue $E_{\kappa}$ equals $E_{\kappa}^{0}$ plus the correction due to the Hughes--Eckart term, eq.~(\ref{eq:Hkappa}).
In case that $V$ represents the Coulomb potential, these two equations are nothing but the separated 
radial Schr\"{o}dinger equations in three dimensions with the inverse power potential of $r^{-\kappa-1}$ type. Namely, one of our main 
goals is to demonstrate that solutions to the three-body Schr\"{o}dinger equation in a Coulomb potential
are approximated by solving a one-dimensional second order ordinary differential equation in $L^{2}(0,\infty;dr)$, 

\begin{equation}
u_{\kappa,l}^{\prime\prime}(r)+\bigl(B_{\kappa}-[l(l+1)/r^{2}]-[A_{\kappa}/r^{\kappa+1}]\bigr)u_{\kappa,l}(r)=0,
\quad l=0,1,\ldots,
\label{eq:uk}
\end{equation}

\noindent{}with some real constants $A_{\kappa}$ and $B_{\kappa}$, the latter being proportional to $E_{\kappa}^{0}$. The subspace 
$\I_{\kappa}$ yields $V_{p}=0$ for all integers $p\geq1$, as $\kappa=0$; in particular, $V_{1}=0$ is a familiar relation known as the 
Wannier saddle point (see eg \cite{Sim87}). Under appropriate boundary conditions functions $u_{0,l}(r)$ are expressed in terms of the 
confluent Whittaker functions, and the associated eigenvalues are proportional to $E_{0}^{0}\propto-A_{0}^{2}/(4n^{2})$, where integers 
$n\geq l+1$. These eigenvalues represent the energies of bound states at threshold, \S\ref{sec:eigenvalues}. Based on perturbative
arguments one deduces that the Hughes--Eckart term does not influence the ground state of $H_{\kappa}$, Proposition~\ref{prop:1};
see also \cite{Sim70}.

Unlike the previous case, admitting $\kappa=1$ we are dealing with operators in the subspace $\A^{1}=[\A,\A]$, for $V_{1}=A_{1}r^{-2}$ is 
the transition potential in the sense of \cite{Fra71} while $V_{p}=0$ for all integers $p\geq2$; $V_{1}\neq0$ ensures a higher accuracy 
in obtaining the eigenvalues $E_{0}^{0}+E_{1}^{0}$ from a series expansion. According to Case \cite{Cas50}, for potentials as singular 
as $r^{-2}$ or greater, there exists a phase factor - proportional to the cut-off radius $r_{0}$ - that describes the breakdown of the 
power law at small distances $r$. Hence $E_{\kappa}=E_{\kappa}(r_{0})$ for $\kappa=1,2,\ldots$ On the other hand, $A_{\kappa}$ is 
proportional to $(-1)^{\kappa}$ (Corollary~\ref{cor:cor4}), which separates attractive potentials from the repulsive ones thus bringing 
in different characteristic aspects of eigenstates.

The task to solve the radial Schr\"{o}dinger equation entailing singular potentials has been of a particular interest for years 
(\cite{Cas50,Fra71,Spe64,Yaf74}) and still it draws the attention of many authors (\cite{Gao08,Gon01,Iqb11,Mor01,Rob00}). If following 
\cite{Cas50} or \cite{Spe64}, there is no ground state for these singular potentials as well as there are only a finite number of bound 
states, as abstracted eg in \cite{Gao98,GaoA99,Gao99}. So one should expect a more valuable contribution to eigenvalue expansion at 
higher energy levels; for more details, see \cite[\S XIII]{Sim78}. Even so, several attempts have been initiated to find the ground state 
energy for a particular class of singular potentials (see eg \cite{Bar80,Nav94}).

The calculation problems of singular potentials are beyond the core of the present paper, though the results on the subject are highly 
appreciated for several reasons (see \S\ref{sec:Coulomb}). We shall not concern ourselves with accomplishing the task to find general 
solution to eq.~(\ref{eq:uk}) but rather demonstrate a method for obtaining it in the Coulomb case; only the eigenvalues $E_{0}$ as well 
as $E_{1}$, in some respects, will be discerned for some illustrative purposes (\S\ref{sec:kappa0}--\ref{sec:numerical}).

\section{\label{sec:Similarity}Similarity for the three-body Hamiltonian operator}

Throughout the whole exposition, we shall exploit several Hilbert spaces of square integrable functions. The typical of them are: 
$L^{2}(\R^{9})$ (as the base space), $L^{2}(\R^{6})$ (as a subspace of translation invariant functions), $L^{2}(\I_{\kappa})$ (as a 
space over vector space $\I_{\kappa}\subset\R^{6}$), and $\h_{\kappa}$ (so that $\h_{\kappa}\otimes\h_{\kappa}\cong 
L^{2}(\I_{\kappa})$). The norm and the scalar (inner) product in a given Hilbert space will be denoted 
by $\Vert\cdot\Vert$ and $(\cdot,\cdot)$; whenever necessary, the subscripts identifying the space will be written as well.

\subsection{\label{subsec:Translation}Translation invariance}

This section summarizes some requisite results obtained from the translation invariance of the potential energy operator $V$.

We say $e(\vect)$ is a representation of the group of translations in $\R^{3}$, denoted E(3), in the space of 
$C^{\infty}(\R^{3})$-functions $f(\vecr_{i})$ if it fulfills $e(\vect)f(\vecr_{i})=f(\vecr_{i}+\vect)$, where vector
$\vect=(t_{x},t_{y},t_{z})$. A Taylor series expansion of $f(\vecr_{i}+\vect)$ yields 
$e(\vect)f(\vecr_{i})=\exp(\mrm{i}\vect\cdot \vecp_{i})f(\vecr_{i})$, where the $\vecp_{i}=-\mrm{i}\nabla_{i}$ 
($i=1,2,3$) are the generators of E(3). Their sum over all $i$ is denoted by $K$.

\begin{lem}\label{lem:Translation}
Define the operator sum $T+V$ by $H$, with $T$ and $V$ as in eqs.~(\ref{eq:T})--(\ref{eq:V}) and assumptions ({\bf A1})--({\bf A3}).
Let $\psi\in\mrm{Ker}(E-H)$ with $E\in\sigma(H)$, where $\psi=\psi(E;\vecr_{1},\vecr_{2},\vecr_{3})$. If $V$ is invariant under the 
action of E(3), then 
(a) functions $\psi$ are translation invariant, one writes $\psi=\varphi(E;\vecr_{12},\vecr_{23},\vecr_{13})$, 
(b) functions $\varphi$ solve the eigenvalue equation for $T_{0}+V^{\prime}$ with the same $E$, namely, 

\begin{subequations}\label{eq:TRANSL}
\begin{align}
&(T_{0}+V^{\prime})\varphi(E)=E\varphi(E), \label{eq:T0} \\
&T_{0}=-\sum_{1\leq i\leq3}(2m_{i})^{-1}\partial_{i}^{2}, \quad\text{on}\quad D(T_{0}) \label{eq:T00}
\end{align}
\end{subequations}

\noindent{}where

$$
\partial_{1}=\nabla_{12}+\nabla_{13},\quad
\partial_{2}=\nabla_{23}-\nabla_{12},\quad
\partial_{3}=-\nabla_{23}-\nabla_{13},
$$

\noindent{}and (c) $D(T_{0})=D(T)$.
\end{lem}

\begin{rem}\label{rem:2}
Here and elsewhere, we distinguish potentials $V$ and $V^{\prime}$ by writing the equation $V_{ij}(\vecr_{i}-\vecr_{j})=
V_{ij}^{\prime}(\vecr_{ij})$. Thus $V\varphi=V^{\prime}\varphi$ but $\nabla_{i}V=\partial_{i}V^{\prime}$. 
\end{rem}

\begin{rem}\label{rem:3}
Vectors $\vecr_{12}$, $\vecr_{23}$, $\vecr_{13}$ are linearly dependent, $\vecr_{12}+\vecr_{23}=\vecr_{13}$, so the 
above given parametrization of $\varphi$ is rather formal, yet a convenient one for our considerations. Hence  $\varphi\in 
D(T)\subset L^{2}(\R^{6})$. We shall regard this aspect once more in \S\ref{sec:Coulomb}.
\end{rem}

\begin{rem}\label{rem:4}
It appears that $(T+V)\psi=(T_{0}+V^{\prime})\varphi$ is the identity. For this reason, we shall define $T_{0}+V^{\prime}$ by the
same $H$; there will be no possibility of confusion.
\end{rem}

\begin{proof}[Proof of Lemma~\ref{lem:Translation}]
The translation invariance of $V$ infers $KV=0$. On the other hand, $[H,\vecp_{i}]=-\vecp_{i}V$. Hence $[H,K]=0$. The three components 
$K_{x}$, $K_{y}$, $K_{z}$ of $K$ commute with each other and thus for our purposes, it suffices to choose one of them, say $K_{z}$.

The commutator $[H,K_{z}]=0$ yields $\psi\in\mrm{Ker}(\lambda-K_{z})\cap\mrm{Ker}(E-H)\neq\varnothing$, with
$\lambda\in\sigma(K_{z})$. Subsequently, functions $\psi$ associated with $E$ are labeled by $\lambda$ as well. One writes 
$\psi=\psi(E\lambda)$. These functions solve the following equation

$$
\frac{\partial\psi}{\partial z_{1}}+\frac{\partial\psi}{\partial z_{2}}+\frac{\partial\psi}{\partial z_{3}}=\mrm{i}\lambda\psi.
$$

\noindent{}The partial differential equation is satisfied whenever functions $\psi$ take one of the following forms:
$\exp(\mrm{i}\lambda z_{1})\varphi$, $\exp(\mrm{i}\lambda z_{2})\varphi$ or $\exp(\mrm{i}\lambda z_{3})\varphi$, where translation 
invariant functions $\varphi$ (with $K_{z}\varphi=0$) are labeled by $E$ only, $\varphi=\varphi(E;\vecr_{12},\vecr_{23},\vecr_{13})$. 
Since all three forms are equivalent (with their appropriate functions $\varphi$), we choose the first one. Note that it suffices to 
choose functions $\varphi$ being invariant under translations along the $z$ axis only. By applying the same procedure for the remaining 
components, $K_{x}$, $K_{y}$, we would deduce that functions $\varphi$ are invariant under translations along the $x$, $y$ axes as well. 
Bearing this in mind, we deduce that functions $\varphi$ are invariant under translations along all three axes associated with each 
$\vecr_{ij}$, that is, in $\R^{3}\times\R^{3}\times\R^{3}$. 

The application of $e(t_{z})$ to $\psi=\exp(\mrm{i}\lambda z_{1})\varphi$ yields $e(t_{z})\psi=\exp(\mrm{i}\lambda t_{z})\psi$. This 
means that functions $\psi$ labeled by a particular E(3)-scalar representation, $\lambda=0$, are translation invariant, 
$\psi(E0)=\varphi(E)$.

We wish to find the operator $T_{\lambda}$ whose range for all $\varphi$ in $D(T_{\lambda})$ is the same as that of 
$T$ for all $\psi$ in $D(T)$, namely, $T_{\lambda}\varphi=T\psi$. First, we calculate the gradients of 
$\exp(\mrm{i}\lambda z_{1})\varphi$. Second, we calculate the corresponding Laplacians. Third, we substitute obtained expressions in 
eq.~(\ref{eq:T}). The result reads

$$
T_{\lambda}\varphi=e^{\mrm{i}\lambda z_{1}}(T_{0}+\lambda t_{\lambda})\varphi\quad \text{with}\quad 
t_{\lambda}\varphi=-(2m_{1})^{-1}\biggl[-\lambda+2\mrm{i}\biggl(\frac{\partial}{\partial z_{12}}+
\frac{\partial}{\partial z_{13}}\biggr)\biggr]\varphi.
$$

\noindent{}For a particular $\lambda=0$, we get the tautology $T_{0}=T_{0}$. [Note:
$\partial_{i}\varphi=\nabla_{i}\psi$.] For arbitrary $\lambda$, substitute $T_{\lambda}$ in 
$(T+V)\psi(E\lambda)=E\psi(E\lambda)$ and get the equation $(T_{0}+V^{\prime}+\lambda t_{\lambda})\varphi(E)=E\varphi(E)$. 
The number of eigenvalues $\lambda$ is infinite, and the latter must hold for all of them.
It follows from $\lambda t_{\lambda}\varphi(E)=0$ that $\lambda=0$ or $t_{\lambda}\varphi(E)=0$. But $t_{\lambda}\varphi(E)=0$
is improper since $\varphi=\varphi(E)$ is $\lambda$ independent. Therefore, $\lambda=0$ and functions $\psi$ are translation invariant, 
$\psi=\varphi(E)$, and they satisfy eq.~(\ref{eq:T0}). This gives items (a)--(b). Item (c) follows immediately due to $\psi=\varphi$. 
The proof is complete.
\end{proof}

\begin{lem}\label{lem:lem1}
Define $\sum_{i<j}\nabla_{ij}$ by $G$, where the sum runs over all $1\leq i<j\leq 3$. Then there exist domains 
$D$, $D^{\prime}\subset D(T_{0})$ such that $D=\{\varphi\in D(T_{0})\co G\varphi=0\}$, $D\cup D^{\prime}\subseteq D(T_{0})$ and 
$D\cap D^{\prime}=\varnothing$.
\end{lem}

\begin{proof}
As above, let us choose the component $G_{z}$. Then the commutators $[G_{z},K_{z}]=0$, $[H,K_{z}]=0$, but $[G_{z},H]\neq0$, in general. 
This indicates that the (only one) eigenvalue $\lambda=0$ (proof of Lemma~\ref{lem:Translation}) of $K_{z}$ is 
degenerate, where degenerate eigenfunctions are $\varphi\in D(T_{0})$ (Remark~\ref{rem:4}). Therefore, if given
$\varphi_{1}\in D_{1}$, $\phi\in D^{\prime}$, $D_{1}\cup D^{\prime}=D(T_{0})$ and $D_{1}\cap D^{\prime}=\varnothing$, then 
$G_{z}\varphi_{1}=\mu \varphi_{1}$ and $H\phi=E\phi$, with some real numbers $\mu$ for all $\varphi_{1}\in D_{1}$ and 
$\phi\in D^{\prime}$. Solutions to $G_{z}\varphi_{1}=\mu \varphi_{1}$ are translation invariant functions $\varphi_{1}$. In turn, 
$\varphi_{1}$ is represented by a certain translation invariant function $\tilde{\varphi}_{1}$ multiplied by either $\exp(\mu z_{12})$ 
or $\exp(\mu z_{23})$ or $\exp(\mu z_{13})$. Subsequently, solutions satisfying $G_{z}\varphi_{1}=\mu \varphi_{1}$ with $\mu=0$ are in 
$D_{1}$ as well, by the proof of Lemma~\ref{lem:Translation}. The nonempty subset $D\subseteq D_{1}$ of these solutions is exactly what 
we were looking for.
\end{proof}

\begin{rem}\label{rem:6}
It follows from the above lemma that $H\varphi=E\varphi$, with $H\co D(T_{0})\to L^{2}(\R^{6})$, actually means $H\phi=E\phi$, with 
$H\co D^{\prime}\to L^{2}(\R^{6})$ (see also Remarks~\ref{rem:3}--\ref{rem:4}). 
\end{rem}

\begin{cor}\label{cor:cor1}
The operator $T_{0}\co D\to L^{2}(\R^{6})$ is represented by the following equivalent forms

\begin{subequations}\label{eq:Tk}
\begin{align}
T_{0}=&
-\alpha\Delta_{12}-\beta\Delta_{23}+\gamma(\nabla_{12}\cdot\nabla_{23}), \label{eq:Tk1} \\
=&
-\xi\Delta_{23}-\alpha\Delta_{13}+\zeta(\nabla_{13}\cdot\nabla_{23}), \label{eq:Tk2} \\
=&
-\xi\Delta_{12}-\beta\Delta_{13}+\eta(\nabla_{12}\cdot\nabla_{13}). \label{eq:Tk3}
\end{align}
\end{subequations}

\noindent{}The real numbers $\alpha$, $\beta$, $\gamma$, $\xi$, $\zeta$, $\eta$ are equal to 

\begin{align}
&\alpha=\frac{1}{2}\biggl(\frac{1}{m_{2}}+\frac{1}{m_{3}}\biggr),\quad
\beta=\frac{1}{2}\biggl(\frac{1}{m_{1}}+\frac{1}{m_{2}}\biggr),\quad
\gamma=\frac{1}{m_{2}}, \nonumber\\
&\xi=\frac{1}{2}\biggl(\frac{1}{m_{1}}+\frac{4}{m_{2}}+\frac{1}{m_{3}}\biggr),\quad
\zeta=-\biggl(\frac{2}{m_{2}}+\frac{1}{m_{3}}\biggr),\quad
\eta=-\biggl(\frac{1}{m_{1}}+\frac{2}{m_{2}}\biggr). \nonumber
\end{align}
\end{cor}

\begin{proof}
Combining Lemmas \ref{lem:Translation}--\ref{lem:lem1}, simply substitute $G\theta=0$ in eq.~(\ref{eq:T00}) for $\theta\in D$, and 
the result follows.
\end{proof}

\begin{rem}\label{rem:rem8}
Supposing that all three $T_{0}$-forms are equivalent, we choose the first one, eq.~(\ref{eq:Tk1}). Note that, in general, 
Corollary~\ref{cor:cor1} does not apply to $T_{0}$, eq.~(\ref{eq:T00}), representing the map from $D^{\prime}$ to $L^{2}(\R^{6})$ 
(Remark~\ref{rem:6}).
\end{rem}

\subsection{\label{subsec:Untary}Unitary equivalence}

Let $V_{n}$, for every integer $n=0,1,\ldots$, denote $G_{z}^{n}V^{\prime}$, where $G_{z}^{n}=G_{z}G_{z}\ldots G_{z}$ ($n$ times); 
clearly, $G_{z}^{0}$ is the identity operator and $V_{0}\equiv V^{\prime}$. Let $\I_{\kappa}$ ($\kappa=0,1,\ldots$),

\begin{equation}
\I_{\kappa}=\{(\vecr_{12},\vecr_{23},\vecr_{13})\in\R^{9}\co 
\vecr_{12}+\vecr_{23}=\vecr_{13},V_{\kappa+p}=0\;\forall p=1,2,\ldots\},
\label{eq:Inu}
\end{equation}

\noindent{}be a nonempty subset in $\R^{6}$ such that for some arbitrary integer $\kappa\geq0$, the $\kappa_{1}^{\text{th}}$ derivative 
of smooth function $V^{\prime}$ were equal to zero for all $\kappa_{1}>\kappa$ (the derivative under consideration is defined in 
Lemma~\ref{lem:lem1}). In what follows, we shall identify the spaces endowed with vectors in $\I_{\kappa}$ by $\kappa$; hence  
$D_{\kappa}$, $D_{\kappa}^{\prime}$ and $D_{0,\kappa}$. Here, domains $D_{\kappa}$ and $D_{\kappa}^{\prime}$ are considered in a similar 
way as $D$ and $D^{\prime}$ in Lemma~\ref{lem:lem1} with $D(T_{0})$ replaced by $D_{0,\kappa}$, whereas $D_{0,\kappa}$ is the set of 
functions from $L^{2}(\I_{\kappa})$ such that: (1) $\Vert T_{0}\varphi\Vert_{\kappa}<\infty$ for all $\varphi\in D_{0,\kappa}$; 
(2) $T_{0}$ is self-adjoint on $D_{0,\kappa}$; (3) $D(V_{n})\supset D_{0,\kappa}$ for all $n=0,1,\ldots$ Here and elsewhere 
$\Vert\cdot\Vert_{\kappa}$ is the $L^{2}(\I_{\kappa})$-norm. By items (1) to (3), the operator sum $T_{0}+V_{n}$, denoted $H_{n}$, is a 
self-adjoint operator in $L^{2}(\I_{\kappa})$ with domain $D_{0,\kappa}$. In particular, $H_{n}=T_{0}$ for all $n=\kappa+1, \kappa+2,
\ldots$

\begin{rem}\label{rem:measure}
We shall clarify the meaning of $L^{2}(\I_{\kappa})$-norm. As it is clear from the definition, eq.~(\ref{eq:Inu}), $L^{2}(\I_{\kappa})$ 
is nothing but $L^{2}(\R^{6}$), with the (Lebesgue) measure $d\mu_{\kappa}$ whose exact form depends on the form of $V^{\prime}$. 
One writes $L^{2}(\I_{\kappa})\equiv L^{2}(\R^{6},d\mu_{\kappa})$ and $L^{2}(\R^{6})\equiv L^{2}(\R^{6},d\mu)$,
where $\mu_{\kappa}=\mu_{\kappa}(\I_{\kappa})$ and
$\mu_{\kappa}(\I_{\kappa})\subset\mu(\R^{6})$. Moreover, given $\kappa=0,1,\ldots$, the measures $\{\mu_{\kappa}\}$ are mutually
singular since they satisfy $\mu_{\kappa}(\R^{6})=\mu_{\kappa}(\I_{\kappa})+\mu_{\kappa}(\R^{6}\backslash\I_{\kappa})=
\mu_{\kappa}(\I_{\kappa})$ provided $\mu_{\kappa}(X\backslash\I_{\kappa})=0$ by eq.~(\ref{eq:Inu}) for any subset 
$\I_{\kappa}\subset X\subset\R^{6}$. Indeed, given $\kappa=0$ and $1$. Define $\I_{*}=\I_{0}\cup\I_{1}$. Then
$\mu_{0}(\I_{*})+\mu_{1}(\I_{*})=\mu_{0}(\I_{0})+\mu_{1}(\I_{*}\backslash\I_{0})$. But $\mu_{1}(\I_{*}\backslash\I_{0})=\mu_{1}(\I_{1})$.
Therefore, for $\mu_{*}=\mu_{0}+\mu_{1}$, $L^{2}(\R^{6},d\mu_{*})$ is isomorphic to $L^{2}(\R^{6},d\mu_{0})\oplus 
L^{2}(\R^{6},d\mu_{1})$, which can be naturally applied to arbitrary $\kappa$. Hence for $\mu^{\prime}=\mu_{0}+\mu_{1}+\ldots$,
$L^{2}(\R^{6},d\mu^{\prime})$ is isomorphic to $L^{2}(\R^{6},d\mu_{0})\oplus L^{2}(\R^{6},d\mu_{1})\oplus\ldots$. Below we shall
demonstrate that $\mu^{\prime}=\mu$.

Finally, we are in a position to define the $L^{2}(\I_{\kappa})$-norm. By assumption $\mu^{\prime}=\mu$, one may write
$\int_{\R^{6}}d\mu=\sum_{\kappa}\int_{\I_{\kappa}}d\mu_{\kappa}$, where the sum runs from $0$ to $\infty$. Equivalently,
$\int_{\I_{\kappa}}d\mu_{\kappa}=\lambda_{\kappa}\int_{\R^{6}}d\mu$ for some nonzero $\lambda_{\kappa}$ such that 
$\sum_{\kappa}\lambda_{\kappa}=1$. With this definition, that the operator $H_{n}$ is in $L^{2}(\I_{\kappa})$, it actually means that
the operator $\lambda_{\kappa}^{-1}H_{n}$ is in $L^{2}(\R^{6},\lambda_{\kappa}d\mu)$, and therefore there is unitary $U$ such that
$U(\lambda_{0}^{-1}H_{0}\oplus\lambda_{1}^{-1}H_{1}\oplus\ldots)U^{-1}$ is in $L^{2}(\R^{6})$.
\end{rem}

The reason for defining a subset $\I_{\kappa}$ is to divide a noncommutative relation $[G_{z},H]$ into the
commutative ones, as it will be shown thereon. The properties of $\I_{\kappa}$ in the case of Coulomb potentials will be assembled
in \S\ref{sec:stability}.

The main goal of the present paragraph is to demonstrate that the eigenvalues of $H$ in $L^{2}(\R^{6})$ can be established from
the eigenvalues of $H_{\kappa}$ in $L^{2}(\I_{\kappa})$.

\begin{thm}\label{thm:main}
Given the operator $[u,v](\varphi)=w(\varphi)$ in $\I_{\kappa}$ for $\kappa=0,1,\ldots$ and $\varphi\in D_{0,\kappa}$. The elements 
$u$, $v$, $w$ denote any operator from $G_{z}$, $V_{0}$, $V_{1}$, $\ldots$, $V_{\kappa}$. Then (1)

\begin{equation}
[G_{z},V_{n}]=V_{n+1},\quad
[V_{n},V_{m}]=0\quad\text{for all}\quad n,m=0,1,2,\ldots,\kappa.
\label{eq:LieAlg}
\end{equation}

\noindent{}The element $V_{n}$ denotes the operator of multiplication by $V_{n}$. 
(2) The commutation relations in eq.~(\ref{eq:LieAlg}) define the Lie algebra, denoted $\A=\A(\I_{\kappa})$, with an operation
$\I_{\kappa}\times \I_{\kappa}\to \I_{\kappa}$.
\end{thm}

\begin{proof}
First, we shall prove eq.~(\ref{eq:LieAlg}). Second, we shall demonstrate that $\A$ is a Lie algebra indeed. Note that eventually
the commutator $[G_{z},V_{n}]$ terminates at $n=\kappa$ due to the definition of $\I_{\kappa}$, eq.~(\ref{eq:Inu}).

(1) Let us calculate the first commutator in eq.~(\ref{eq:LieAlg}), namely,

\begin{align*}
[G_{z},V_{n}](\varphi)=&
G_{z}V_{n}(\varphi)-V_{n}G_{z}(\varphi)=G_{z}(V_{n}\varphi)-V_{n}(G_{z}\varphi) \\
=&(G_{z}V_{n})(\varphi)=(G_{z}^{n+1}V^{\prime})(\varphi).
\end{align*}

\noindent{}But $(G_{z}^{n+1}V^{\prime})(\varphi)=V_{n+1}(\varphi)$, by definition. Therefore, the identity is immediate. The second
commutator, $[V_{n},V_{m}](\varphi)=V_{n}V_{m}(\varphi)-V_{m}V_{n}(\varphi)=0$, is evident provided $V_{n}$ and $V_{m}$ represent 
numerical functions.

(2) Elements $u=G_{z}$, $V_{0}$, $V_{1}$, $\ldots$, $V_{\kappa}$ form the basis of a linear space $\A(\I_{\kappa})$ of dimension 
$\kappa+2$. If endowed with the binary operation $\A\times\A\to\A$, denoted $(u,v)\mapsto[u,v]$ for all
$u$, $v$ in $\A$, the linear space $\A$ must fulfill the bilinearity, anticommutativity and Jacobi identity. 

Bilinearity: $[au+bv,w]=a[u,w]+b[v,w]$ for any scalars $a$, $b$. Due to the commutativity 
$au+bv=bv+au$ and distributivity $au+bu=(a+b)u$, it suffices to consider two cases: 
(a) $u=G_{z}$, $v=V_{n}$, $w=V_{m}$;
(b) $u=G_{z}$, $v=V_{n}$, $w=G_{z}$; in all cases $n$, $m=0,1,\ldots,\kappa$.

\begin{align*}
\text{(a): }& [aG_{z}+bV_{n},V_{m}](\varphi)=(aG_{z}+bV_{n})V_{m}(\varphi)-V_{m}(aG_{z}+bV_{n})(\varphi) \\
&=aG_{z}V_{m}(\varphi)+bV_{n}V_{m}(\varphi)-aV_{m}G_{z}(\varphi)-bV_{m}V_{n}(\varphi) \\
&= a[G_{z},V_{m}](\varphi)+b[V_{n},V_{m}](\varphi) \\
\text{(b): }& [aG_{z}+bV_{n},G_{z}](\varphi)=(aG_{z}+bV_{n})G_{z}(\varphi)-G_{z}(aG_{z}+bV_{n})(\varphi) \\
&=aG_{z}G_{z}(\varphi)+bV_{n}G_{z}(\varphi)-aG_{z}G_{z}(\varphi)-bG_{z}V_{n}(\varphi) \\
&= a[G_{z},G_{z}](\varphi)+b[V_{n},G_{z}](\varphi)
\end{align*}

\noindent{}Hence $[\cdot,\cdot]$ is bilinear.

Anticommutativity: $[u,v]=-[v,u]$; in particular, $[u,u]=0$. This property is easy to verify by using
the distributivity of the addition operation. Two cases are considered: 
(a) $u=G_{z}$, $v=V_{n}$;
(b) $u=V_{n}$, $v=V_{m}$; in all cases $n$, $m=0,1,\ldots,\kappa$.

\begin{align*}
\text{(a): }& [G_{z},V_{n}](\varphi)=G_{z}V_{n}(\varphi)-V_{n}G_{z}(\varphi)=
-(V_{n}G_{z}-G_{z}V_{n})(\varphi)=-[V_{n},G_{z}](\varphi) \\
\text{(b): }& [V_{n},V_{m}](\varphi)=V_{n}V_{m}(\varphi)-V_{m}V_{n}(\varphi)=
-(V_{m}V_{n}-V_{n}V_{m})(\varphi)=-[V_{m},V_{n}](\varphi)
\end{align*}

\noindent{}Hence $[\cdot,\cdot]$ is anticommutative.

Jacobi identity: $[[u,v],w]+[[v,w],u]+[[w,u],v]=0$. The identity is antisymmetric with
respect to the permutation of any two elements. Thus it suffices to choose $u=G_{z}$, $v=V_{n}$, $w=V_{m}$
($n$, $m=0,1,\ldots,\kappa$). Then applying eq.~(\ref{eq:LieAlg}) and anticommutativity of $[\cdot,\cdot]$ one finds that

\begin{align*}
&[[G_{z},V_{n}],V_{m}](\varphi)+[[V_{n},V_{m}],G_{z}](\varphi)+[[V_{m},G_{z}],V_{n}](\varphi)=
[G_{z},V_{n}]V_{m}(\varphi) \\
&-V_{m}V_{n+1}(\varphi)+[V_{n},V_{m}]G_{z}(\varphi)-[G_{z},V_{m}]V_{n}(\varphi)+V_{n}V_{m+1}(\varphi) \\
&= [V_{n+1},V_{m}](\varphi)+[V_{n},V_{m+1}](\varphi)=0.
\end{align*}

\noindent{}This completes the proof.
\end{proof}

\begin{cor}\label{cor:cor2}
The Lie algebra $\A(\I_{\kappa})$ is nilpotent with the nilpotency class $\kappa+1$.
\end{cor}

\begin{proof}
We need only to compute the length of the lower central series containing ideals $\A^{c}=[\A,\A^{c-1}]$, where $\A^{0}=
\A$, $c\geq1$. Assume that $c=1$. Then $\A^{1}=[\A,\A]$ is spanned by the elements $[u,v]$, where
$u$, $v$ are in $\A$. Thus by eq.~(\ref{eq:LieAlg}), the first commutator subalgebra $\A^{1}$ contains 
elements $\{V_{n}\}_{n=1}^{\kappa}$. Similarly, the second subalgebra ($c=2$) $\A^{2}=[\A,\A^{1}]$ is spanned by the 
elements $[u,v]$, where $u\in\A$ and $v\in\A^{1}$, hence $\{V_{n}\}_{n=2}^{\kappa}$ etc. 
Finally, $\A^{\kappa+1}=0$. This proves that $\A$ is nilpotent with the central series of length $\kappa+1$.
\end{proof}

The nilpotency of $\A(\I_{\kappa})$ implies the existence of an isomorphism from $\A$ to the Lie algebra of strictly
upper-triangular matrices (\cite[\S I.3]{Hum72}; \cite[\S 3.1.6]{Zhe83}). Equivalently, there exists a 
representation $\varrho\co\A\to\mathfrak{gl}(\I_{\kappa})$ given by $\varrho(e)=G_{z}$, $\varrho(f_{n})=V_{n}$ ($n=0,1,\ldots,\kappa$), 
where elements $g=e,f_{n}\in\A$ denote the strictly upper-triangular matrices; in particular, $\A(\I_{1})$ is isomorphic 
to the Heisenberg algebra, whereas $\A(\I_{0})$ is commutative. Clearly, $[e,f_{n}]=f_{n+1}$, $[f_{n},f_{m}]=0$ for all 
$n,m=0,1,\ldots,\kappa$, and $f_{\kappa+p}=0$ for all $p\geq1$. 

Assume now that $\mathcal{L}$ is a matrix Lie group with Lie algebra $\A$ and $\exp\co\A\to\mathcal{L}$ is the exponential mapping
for $\mathcal{L}$. Then the $\exp(\mrm{i}tg)$ are in $\mathcal{L}$ for all $g\in\A$ and for all $t\in\R$. Provided
$\Pi\co\mathcal{L}\to\mathtt{GL}(\I_{\kappa})$ is a representation of $\mathcal{L}$ on $\I_{\kappa}$, we get that 
$\Pi\bigl(\exp(\mrm{i}g)\bigr)=\exp\bigl(\mrm{i}\varrho(g)\bigr)$. But $\varrho(e)=G_{z}$, and thus 
$\Pi\bigl(\exp(\mrm{i}e)\bigr)(\theta)=I(\theta)$, the identity for all $\theta\in D_{\kappa}$, by Lemma~\ref{lem:lem1}. On the 
other hand, $\varrho(f_{n})=V_{n}$ ($n=0,1,\ldots,\kappa$) and $\Pi\bigl(\exp(\mrm{i}f_{n})\bigr)(\theta)=
\exp(\mrm{i}V_{n})(\theta)$. The elements $I$, $\{\exp(\mrm{i}tV_{n})\}_{n=0}^{\kappa}$ therefore form a (bounded) group
of unitary operators given by the map $D_{\kappa}\to L^{2}(\I_{\kappa})$ for all $t\in\R$. In turn, 
it is a subgroup of the group generated by $\mrm{i}H_{n}$, where
$H_{n}=T_{0}+V_{n}$ and $T_{0}\co D_{0,\kappa}\to L^{2}(\I_{\kappa})$ is self-adjoint.
Due to the filtering $\A\supset\A^{1}\supset\ldots\supset\A^{\kappa}$, the elements of a set 
$\{H_{n}\}_{n=0}^{\kappa}$ converge to a single $H_{\kappa}$ [recall that the kernel of a Lie algebra 
homomorphism $\A^{c}\to\A^{c+1}$ is $\{f_{c}\}$], which in turn commutes with $G_{z}$, namely,
$[G_{z},H_{\kappa}]=0$, by Corollary~\ref{cor:cor2}. As a result, the eigenfunctions of operator $H_{\kappa}\co D_{\kappa}\to
L^{2}(\I_{\kappa})$ are those of $G_{z}$, and thus $\theta\in\mrm{Ker}(E_{\kappa}-H_{\kappa})\cap\mrm{Ker}(0-G_{z})\neq\varnothing$
for $E_{\kappa}\in\sigma(H_{\kappa})$, by Lemmas~\ref{lem:Translation}--\ref{lem:lem1}. In particular, whenever $\I_{0}$ is nonempty, 
one should expect that $E_{0}=E$ due to the formal coincidence of $H_{0}$ with $H$ (Remark~\ref{rem:6}). However, $H$ is in 
$L^{2}(\R^{6})$ and it is defined on $D^{\prime}$ whereas $H_{0}$ in $L^{2}(\I_{\kappa})$ is defined on $D_{\kappa}$ at a
particular value $\kappa=0$ (see also Remark~\ref{rem:rem8}). This means $E_{0}\neq E$, in general (that is, for smooth $V^{\prime}$). 
On the other hand, provided $\I_{\kappa}$ is nonempty for arbitrary $\kappa$, one finds from the above considered Lie algebra filtering 
that $D^{\prime}$ is the space decomposition $\oplus_{\kappa}D_{\kappa}$, where $\kappa$ goes from $0$ to $\infty$. But $D^{\prime}$
is dense in $L^{2}(\R^{6})$ and $D_{\kappa}$ in $L^{2}(\I_{\kappa})$. Thus by Remark~\ref{rem:measure}, $U$ is the unitary transformation
from $\oplus_{\kappa}L^{2}(\I_{\kappa})$ to $L^{2}(\R^{6})$ so that $U(\lambda_{1}^{-1}H_{1}\oplus\lambda_{2}^{-1}H_{2}\oplus\ldots)
U^{-1}=H$. Subsequently, for $E\in\sigma(H)$ and $E_{\kappa}\in\sigma(H_{\kappa})$, one finds that $E=\sum_{\kappa}c_{\kappa}E_{\kappa}$,
where $c_{\kappa}=\lambda_{\kappa}^{-1}(\Vert\theta\Vert_{\kappa}/\Vert\phi\Vert)^{2}$. But $\sum_{\kappa}\lambda_{\kappa}=1$ and
$\Vert\phi\Vert^{2}=\sum_{\kappa}\Vert\theta\Vert_{\kappa}^{2}$, with $\theta=\theta(E_{\kappa})$. One thus derives
$\lambda_{\kappa}=(\Vert\theta\Vert_{\kappa}/\Vert\phi\Vert)^{2}$ and

\begin{equation}
E=\sum_{\kappa=0}^{\infty}E_{\kappa}.
\label{eq:E}
\end{equation}

\noindent{}As a result, we have established that solutions to the initially admitted eigenvalue equation $H\phi=E\phi$ 
in $L^{2}(\R^{6})$ are obtained by solving $H_{\kappa}\theta=E_{\kappa}\theta$ in $L^{2}(\I_{\kappa})$, where 
$H_{\kappa}=T_{0}+V_{\kappa}$ with $T_{0}$ given in eq.~(\ref{eq:Tk1}) and $\theta\in D_{\kappa}$. 

In the next section, we shall be concerned with the Coulomb potentials, though one can easily enough apply the method to be
presented to other spherically symmetric potentials imposed under ({\bf A1})--({\bf A3}).

\section{\label{sec:Coulomb}Solutions for the three-body Hamiltonian operator with Coulomb potentials}

The Coulomb potential $V^{\prime}$ is a spherically symmetric translation invariant function represented as a sum of functions
$V_{ij}^{\prime}=Z_{ij}/r_{ij}$, eq.~(\ref{eq:V}); for the notations exploited here, recall Remark~\ref{rem:2}. The scalar 
$Z_{ij}=Z_{ji}=Z_{i}Z_{j}$, where $Z_{i}$ ($i=1,2,3$) denotes a nonzero integer (the charge of the $i^{\text{th}}$ particle).
The spherical symmetry in $\R^{3}$ preserves rotation invariance under SO(3) thus simplifying the Laplacian $\Delta_{ij}=
d^{2}/dr_{ij}^{2}+(2/r_{ij})d/dr_{ij}-l(l+1)/r_{ij}^{2}$, where $l$ labels the SO(3)-irreducible representation. Bearing in mind
Remarks~\ref{rem:3}--\ref{rem:rem8}, we shall use $l_{1}$ to label representations for $ij=12$, and $l_{2}$ for $ij=23$; the associated
basis indices will be identified by $\pi_{1}=-l_{1},-l_{1}+1,\ldots,0,1,\ldots,l_{1}$ and by 
$\pi_{2}=-l_{2},-l_{2}+1,\ldots,0,1,\ldots,l_{2}$, respectively.

\subsection{\label{sec:stability}The stability criterion}

Let us first study the properties of a subset $\I_{\kappa}$ introduced in \S\ref{subsec:Untary}. Proceeding from the definition, 
eq.~(\ref{eq:Inu}), we deduce that the nilpotency of $\A$ is ensured whenever (see also the proof of Corollary~\ref{cor:cor4})

\begin{equation}
\sum_{1\leq i<j\leq3}\frac{Z_{ij}}{r_{ij}^{k}}=0\quad\text{for all}\quad k=\kappa+p+1=2,3,\ldots\quad\text{for all}\quad p=1,2,\ldots
\label{eq:zeroV}
\end{equation}

\noindent{}provided the $z$ axis is suitably oriented. Equation~(\ref{eq:zeroV}) suggests that at least one integer $Z_{i}$ from 
$Z_{1}$, $Z_{2}$, $Z_{3}$ must be of opposite sign -- this is what we call the stability criterion for the Coulomb three-body system 
(see also Remark~\ref{rem:17}). There is the classical picture to it: If the three particles are all negatively (positively) charged,
they move off from each other to infinity due to the Coulomb repulsion. A well-known observation follows therefore from
the requirement that the Lie algebra $\A$ were nilpotent. Henceforth, we accept the criterion validity.

Linearly dependent vectors $\{\vecr_{ij}\}$ (Remark~\ref{rem:3}) form a triangle embedded in $\R^{3}$. Based on the present condition, 
we can prove the following result.

\begin{lem}\label{lem:criter}
Let $\omega_{k}$, $\sigma_{k}$, $\tau_{k}$ denote the angles between the pairs of vectors $(\vecr_{12}$, $\vecr_{13})$,
$(\vecr_{13}$, $\vecr_{23})$, $(\vecr_{12}$, $\vecr_{23})$, respectively. If a given three-body system is stable, then for 
any integer $k\geq2$ such that (i) $(Z_{2}/Z_{3})+(Z_{2}/Z_{1})\wp^{k}<0$ and (ii)

\begin{align*}
C_{1}(\wp)\leq&\Biggl(-\frac{Z_{13}}{Z_{12}+Z_{23}\wp^{k}}\Biggr)^{1/k}\leq\frac{1+\wp}{\wp},\quad
C_{1}(\wp)=\left\{\!\!\!\begin{array}{ll}
\frac{1-\wp}{\wp},& 0<\wp\leq1, \\
0, & \wp>1,
\end{array}\right.
\intertext{($\sigma_{k}$ is acute)}
0\leq&\Biggl(-\frac{Z_{13}}{Z_{12}+Z_{23}\wp^{k}}\Biggr)^{1/k}\leq C_{2}(\wp),\quad
C_{2}(\wp)=\left\{\!\!\!\begin{array}{ll}
0, & 0<\wp\leq1, \\
\frac{\wp-1}{\wp}, & \wp>1,
\end{array}\right.
\end{align*}

\noindent{}($\sigma_{k}$ is obtuse),
there exists a multiplier $\wp\geq0$ satisfying $\sin\sigma_{k}=\wp\sin\omega_{k}$ so that eq.~(\ref{eq:zeroV}) 
holds for all $0\leq\omega_{k},\sigma_{k}\leq\pi$ such that: 
(1) if $0\leq\omega_{k},\sigma_{k}$ $\leq\pi/2$, then $0\leq \wp\leq1/\sin\omega_{k}$; 
(2) if $0\leq\omega_{k}\leq\pi/2$ and $\pi/2<\sigma_{k}\leq\pi$, then $1<\wp\leq1/\sin\omega_{k}$; 
(3) if $\pi/2<\omega_{k}\leq\pi$ and $0\leq\sigma_{k}\leq\pi/2$, then $0\leq \wp<1$; 
(4) if $\pi/2<\omega_{k},\sigma_{k}\leq\pi$, then $\wp$ does not exist.
In case that $Z_{1}/Z_{3}<0$, the multiplier $\wp\neq(-Z_{1}/Z_{3})^{1/k}$ for all suitable $0\leq\omega_{k},\sigma_{k}\leq\pi$.
\end{lem}

\begin{rem}\label{rem:13}
In particular, lemma states that for a certain integer $k\geq2$, if such exists at all, one can find a multiplier $\wp\geq0$ such that
the angles $\omega_{k}$, $\sigma_{k}$, $\tau_{k}$ obtained from relations
$\sin\sigma_{k}=\wp\sin\omega_{k}$ and $\sin\tau_{k}=c_{k}\sin\omega_{k}$ (where $\tau_{k}=\omega_{k}+\sigma_{k}$) solve
eq.~(\ref{eq:zeroV}). The multiplier $c_{k}=\bigl(-Z_{13}/(Z_{12}+Z_{23}\wp^{k})\bigr)^{1/k}\wp$. Clearly, one should bring to
mind the sine law relating angles with the associated sides of a triangle $\vecr_{12}+\vecr_{23}=\vecr_{13}$.
\end{rem}

\begin{cor}\label{cor:cor3}
Let

$$
D_{k}(\wp)=\bigl\{0\leq\omega_{k},\sigma_{k},\tau_{k}\leq\pi\co 
\omega_{k}+\sigma_{k}-\tau_{k}=0,\sin\sigma_{k}=\wp\sin\omega_{k},\sin\tau_{k}=c_{k}\sin\omega_{k}\bigr\}.
$$

\noindent{}The set $D_{k}(\wp)$ is nonempty if $k\geq2$, the triplet ($1$, $\wp$, $c_{k}$) fulfills the triangle validity, and 
(1) $Z_{1},Z_{2}<0$, $Z_{3}>0$ or $Z_{1},Z_{2}>0$, $Z_{3}<0$ and $\wp^{k}<-Z_{1}/Z_{3}$ or  
(2) $Z_{1}>0$, $Z_{2},Z_{3}<0$ or $Z_{1}<0$, $Z_{2},Z_{3}>0$ and $\wp^{k}>-Z_{1}/Z_{3}$ or 
(3) $Z_{1},Z_{3}>0$, $Z_{2}<0$ or $Z_{1},Z_{3}<0$, $Z_{2}>0$ and $\wp>0$. Otherwise, $D_{k}(\wp)=\varnothing$.
\end{cor}

\begin{proof}[Proof of Lemma~\ref{lem:criter}]
Although the proof to be produced fits any positive integer $k$, we shall make a stronger restrictive condition, $k\geq2$, due to
eq.~(\ref{eq:zeroV}). 

The combination of the sine law, $r_{23}^{-1}\sin\omega_{k}=r_{12}^{-1}\sin\sigma_{k}=r_{13}^{-1}\sin\tau_{k}$, and
eq.~(\ref{eq:zeroV}) points to the following equation

$$
Z_{12}+Z_{23}\frac{\sin^{k}\sigma_{k}}{\sin^{k}\omega_{k}}+Z_{13}\frac{\sin^{k}\sigma_{k}}{\sin^{k}\tau_{k}}=0.
$$

\noindent{}[Note that the values $\omega_{k},\tau_{k}=0,\pi$ are allowed as well by implying $\sigma_{k}=0,\pi$.] Then
the expression for $c_{k}$ (refer to Remark~\ref{rem:13}) follows immediately if $\sin\sigma_{k}=\wp\sin\omega_{k}$ ($\wp\geq0$). The
quantity in parantheses $()^{1/k}$ in $c_{k}$ is positive definite and thus item (i) follows as well. Clearly, the denominator is 
nonzero; otherwise $Z_{1}/Z_{3}<0$ and $\wp\neq(-Z_{1}/Z_{3})^{1/k}$ must hold.

By noting that $0<r_{13}/r_{12}=\cos\omega_{k}\pm(1/\wp^{2}-\sin^{2}\omega_{k})^{1/2}$, we find from eq.~(\ref{eq:zeroV})

$$
Z_{12}+Z_{23}\wp^{k}+\frac{Z_{13}}{\bigl(\cos\omega_{k}\pm(1/\wp^{2}-\sin^{2}\omega_{k})^{1/2}\bigr)^{k}}=0,
$$

\noindent{}where "$+$" is for $0\leq\sigma_{k}\leq\pi/2$, and "$-$" for $\pi/2<\sigma_{k}\leq\pi$. Items (1)--(4) follow
directly from the above equation: Eg let $\omega_{k}$, $\sigma_{k}>\pi/2$. The denominator  is of the form
$(-x)^{1/k}$, $x>0$, hence improper (item (4) in lemma).

Substitute $\tau_{k}=\omega_{k}+\sigma_{k}$ and $\sin\sigma_{k}=\wp\sin\omega_{k}$ in $\sin\tau_{k}=c_{k}\sin\omega_{k}$ and get

$$
c_{k}\sin\omega_{k}=\sin\bigl(\pm\omega_{k}+\arcsin(\wp\sin\omega_{k})\bigr)
$$

\noindent{}[here, again, "$+$" is for $0\leq\sigma_{k}\leq\pi/2$, and "$-$" for $\pi/2<\sigma_{k}\leq\pi$] yielding the estimates
$1-\wp\leq c_{k}\leq 1+\wp$ for acute $\sigma_{k}$, and $0\leq c_{k}\leq\wp-1$ for obtuse $\sigma_{k}$. 
Provided $c_{k}\geq0$, substitute the definition for $c_{k}$ in obtained inequalities and get item (ii). This completes the proof.
\end{proof}

\begin{proof}[Proof of Corollary~\ref{cor:cor3}]
Items (1)--(3) are obvious due to item (i) of Lemma~\ref{lem:criter}. It remains to demonstrate the triangle validity for $1$, $\wp$, 
$c_{k}$. This is done by solving the equation

$$
c_{k}=\wp\biggl(\cos\omega_{k}\pm\sqrt{1/\wp^{2}-\sin^{2}\omega_{k}}\biggr)=
\wp\Biggl(-\frac{Z_{13}}{Z_{12}+Z_{23}\wp^{k}}\Biggr)^{1/k}
$$

\noindent{}which yields

\begin{equation}
\sin\omega_{k}=\frac{\bigl( (1+c_{k}+\wp)(1+c_{k}-\wp)(1-c_{k}+\wp)(c_{k}+\wp-1)\bigr)^{1/2}}{2\wp c_{k}}
\label{eq:wk}
\end{equation}

\noindent{}and hence the triangle validity for the triplet $(1,\wp,c_{k})$ must hold due to inequality $0\leq\sin\omega_{k}\leq1$.
\end{proof}

\begin{cor}\label{cor:cor4}
The spherically symmetric functions $V_{\kappa}$ can be represented by three equivalent forms, where two of them are
$V_{1}^{\kappa}=Z_{12}(\wp;\kappa,k)/r_{12}^{\kappa+1}$, $V_{2}^{\kappa}=Z_{23}(\wp;\kappa,k)/r_{23}^{\kappa+1}$ with

$$
Z_{12}(\wp;\kappa,k)=(-1)^{\kappa}\kappa!\biggl(Z_{12}+Z_{23}\wp^{\kappa+1}+Z_{13}\Bigl(\frac{\wp}{c_{k}}\Bigr)^{\kappa+1}\biggr),
$$

\noindent{}and $Z_{23}(\wp;\kappa,k)=Z_{12}(\wp;\kappa,k)/\wp^{\kappa+1}$.
\end{cor}

\begin{proof}
Differentiate $V_{ij}^{\prime}$ $\kappa$ times with respect to $r_{ij}$,

$$
\frac{d^{\kappa}}{d r_{ij}^{\kappa}}\frac{Z_{ij}}{r_{ij}}=\frac{(-1)^{\kappa}\kappa!Z_{ij}}{r_{ij}^{\kappa+1}}.
$$

\noindent{}Then

$$
V_{\kappa}=\sum_{1\leq i<j\leq 3}\frac{d^{\kappa}}{d r_{ij}^{\kappa}}V_{ij}^{\prime}=(-1)^{\kappa}\kappa!
\sum_{1\leq i<j\leq 3}\frac{Z_{ij}}{r_{ij}^{\kappa+1}}.
$$

\noindent{}Put into use the sine law and Lemma~\ref{lem:criter}, 

$$
\frac{r_{12}}{r_{23}}=\wp,\quad\frac{r_{12}}{r_{13}}=\frac{\wp}{c_{k}}.
$$

\noindent{}Substitute the above equations in $V_{\kappa}$ and get the result.
\end{proof}

\begin{rem}\label{rem:15}
Lemma~\ref{lem:criter} and Corollary~\ref{cor:cor3} provide sufficient information to find nonempty subsets $\I_{\kappa}$. Indeed,
consider given nonzero integers $Z_{i}$ ($i=1,2,3$) and the real, yet unspecified, multiplier $\wp\geq0$. First, establish possible 
integers $k\geq2$, by Lemma~\ref{lem:criter}. Second, substitute determined values of $k$ in $c_{k}$. Third, substitute obtained 
coefficients $c_{k}$ in eq.~(\ref{eq:wk}) and get possible angles $\omega_{k}\in D_{k}(\wp)$ (alternatively, simply apply 
Corollary~\ref{cor:cor3}); the subset $\I_{\kappa}$ is nonempty whenever $D_{k}(\wp)$ is. Applications to some physical systems will be 
displayed in \S\ref{sec:numerical}.
\end{rem}

\begin{rem}\label{rem:16}
Although conditions in Lemma~\ref{lem:criter} and Corollary~\ref{cor:cor3} are invariant under the interchange of integers $Z_{i}$, 
there might appear some arrangement that
does not satisfy Lemma~\ref{lem:criter}. If this is the case, one should select another one. For example, $Z_{1}=Z_{2}=-1$, $Z_{3}=+1$, 
$\wp=1$ brings in $Z_{1}/Z_{3}<0$ and $\wp\neq1$ (see Lemma~\ref{lem:criter}), which contradicts the initially defined $\wp=1$. On the 
other hand, $Z_{1}=Z_{3}=-1$, $Z_{2}=+1$, $\wp=1$ brings in $c_{k}=2^{-1/k}$, and all conditions in Lemma~\ref{lem:criter} as well as 
in Corollary~\ref{cor:cor3} are fulfilled. However, if none of arrangements of $Z_{i}$ fulfill the lemma, one should conclude that the 
three-body system is unstable.
\end{rem}

\begin{rem}\label{rem:17}
We point out that in the present discussion, the definition for stability differs from that exploited by
\cite{Che90,Fro92,Hill77}. Here, we do not study the cases of stability against dissociation (see also 
\cite{Mar92,Ric93}) by assuming these conditions are fulfilled whenever bound states are considered. On the other hand,
provided the three-body system is subjected to the stability criterion, eq.~(\ref{eq:zeroV}), one should deduce from 
Lemma~\ref{lem:criter} that at least $\I_{0}\neq\varnothing$. As an important example of unbound three-body system consider 
the positron-hydrogen system. Our calculated first excited energy [substitute, in atomic units, $Z_{1}=Z_{3}=+1$, $Z_{2}=-1$, 
$m_{1}=m_{2}=1$, $m_{3}=1836.1527$, $n_{1}=2$, $n_{2}=1$ in eq.~(\ref{eq:inft2}) and then convert the result into Rydberg units] equals
$E_{0}\simeq-0.25$ Ry, which is almost consistent with that of Kar and Ho's \cite{Kar05} (see also \cite{Doo78}) derived $S$-wave 
resonance energy (around $-0.257$ Ry) associated with the hydrogen $n=2$ threshold. In the positron-hydrogen system, $E_{1}$ does not 
affect the total energy $E$ when the first excited states are considered, which means some higher eigenstates $E_{\kappa}$ 
($\kappa\geq2$), if such exist, should be included in order to obtain more accurate energies.
\end{rem}

\subsection{\label{sec:eigenvalues}Eigenstates}

We wish to evaluate $E\in\disc(H)$. By eq.~(\ref{eq:E}), $E$ is the sum of $E_{\kappa}\in\disc(H_{\kappa})$, where 

\begin{align}
&H_{\kappa}=H_{\kappa}^{0}+\gamma(\nabla_{12}\cdot\nabla_{23}),\quad
H_{\kappa}^{0}=T_{1}+T_{2}+V^{\prime} \nonumber \\
&\text{and}\quad T_{1}=-\alpha\Delta_{12},\quad T_{2}=-\beta\Delta_{23}
\label{eq:Hkappa}
\end{align}

\noindent{}($\alpha$, $\beta$ and $\gamma$ are as in Corollary~\ref{cor:cor1}). We first consider the Hughes--Eckart term.
Following \cite[Appendix~2]{Sim70}, we demonstrate that:

\begin{prop}\label{prop:1}
$\inf\sigma(H_{\kappa})=\inf\sigma(H_{\kappa}^{0})$.
\end{prop}

\begin{proof}
In agreement with Corollary~\ref{cor:cor4} consider $H_{\kappa}$ in $\vecp$-space

$$
H_{\kappa}= h_{\kappa}+\beta \vecp_{23}^{2}-\gamma(\vecp_{12}\cdot\vecp_{23}), \quad
h_{\kappa}= \alpha \vecp_{12}^{2}+Z_{12}(\wp;\kappa,k)r_{12}^{-\kappa-1}.
$$

\noindent{}[Note that $Z_{12}(\wp;\kappa,k)r_{12}^{-\kappa-1}$ can be replaced by $Z_{23}(\wp;\kappa,k)r_{23}^{-\kappa-1}$; see
Corollary~\ref{cor:cor4}. Subsequently, $\alpha\vecp_{12}^{2}$ is replaced by $\beta\vecp_{23}^{2}$ in $h_{\kappa}$]. Since $\beta>0$, 
and $h_{\kappa}$ and $\beta\vecp_{23}^{2}$ involve independent coordinates, we see $\inf\sigma(h_{\kappa})=\inf\sigma(H_{\kappa}^{0})$.

Let $\vecp=a\vecp_{12}+b\vecp_{23}$ with $b>a>0$. Then $H_{\kappa}^{\prime}=H_{\kappa}+\mu\vecp^{2}$ with $\mu>0$, where
$H_{\kappa}^{\prime}=H_{\kappa}^{\prime\prime}+(\mu a b-\gamma/2)\bm{q}^{2}$, where $\bm{q}=\vecp_{12}+\vecp_{23}$ and

$$
H_{\kappa}^{\prime\prime}=h_{\kappa}+[\mu a(a-b)+\gamma/2]\vecp_{12}^{2}+[\mu b(b-a)+\gamma/2]\vecp_{23}^{2}.
$$

\noindent{}We choose $\mu=\gamma/[2a(b-a)]>0$ for $b>a>0$. Then $\mu a b-\gamma/2$ equals $\gamma a/[2(b-a)]>0$ and

$$
H_{\kappa}^{\prime\prime}=h_{\kappa}+[\beta+\gamma(1+b/a)/2]\vecp_{23}^{2}.
$$

\noindent{}But then $\inf\sigma(H_{\kappa}^{\prime\prime})=\inf\sigma(h_{\kappa})$ since $h_{\kappa}$ and $\vecp_{23}^{2}$ involve
independent coordinates. Subsequently, $\inf\sigma(H_{\kappa}^{\prime})=\inf\sigma(H_{\kappa}^{\prime\prime})$, for 
$\mu a b-\gamma/2>0$, and finally, $\inf\sigma(H_{\kappa}^{\prime})=\inf\sigma(H_{\kappa})$. Hence $\inf\sigma(H_{\kappa}^{0})=
\inf\sigma(H_{\kappa})$ as desired.
\end{proof}

\begin{rem}
Proposition~\ref{prop:1} tells us that the ground state of $H_{\kappa}$ is that of $H_{\kappa}^{0}$.
\end{rem}

Second, consider $H_{\kappa}^{0}$. It is a self-adjoint operator on $D_{0,\kappa}$ whose eigenfunctions are in $D_{\kappa}$.
By Remark~\ref{rem:measure}, $L^{2}(\I_{\kappa})\equiv L^{2}(\R^{6},d\mu_{\kappa})$. But, on the other hand, $L^{2}(\I_{\kappa})$
is isomorphic to $L^{2}(\R^{3},d\mu_{\kappa,1})\otimes L^{2}(\R^{3},d\mu_{\kappa,2})$ provided $d\mu_{\kappa}=d\mu_{\kappa,1}\otimes
d\mu_{\kappa,2}$ \cite[Theorem~II.10]{Sim78I}. We denote $L^{2}(\R^{3},d\mu_{\kappa,i})$ by $\h_{\kappa}$ for $i=1,2$.
Thus there exist unitary operators $U_{1}$ and $U_{2}$ such that

\begin{align}
&U_{1}H_{\kappa}^{0}U_{1}^{-1}=H_{\kappa,1}^{0}\otimes I+I\otimes T_{2},\quad
U_{2}H_{\kappa}^{0}U_{2}^{-1}=T_{1}\otimes I+I\otimes H_{\kappa,2}^{0} \nonumber \\
&\text{with}\quad H_{\kappa,i}^{0}=T_{i}+V_{i}^{\kappa}\quad(i=1,2)
\label{eq:Hkappa1}
\end{align}

\noindent{}and $V_{i}^{\kappa}$ is as in Corollary~\ref{cor:cor4} for $i=1,2$. But then

\begin{equation}
\disc(H_{\kappa}^{0})=\disc(H_{\kappa,1}^{0})=\disc(H_{\kappa,2}^{0})
\label{eq:spec0}
\end{equation}

\noindent{}since $\disc(T_{1})=\disc(T_{2})=\varnothing$. Equation~(\ref{eq:spec0}) allows one to determine $\wp$ 
(Lemma~\ref{lem:criter}) for a given $\kappa$ as well as $E_{\kappa}^{0}\in\disc(H_{\kappa}^{0})$.

Indeed, an ordinary decomposition of product $L^{2}(0,\infty;r^{2}dr)\otimes L^{2}(S^{2})$, with $r=r_{12}$ for $i=1$ and 
$r=r_{23}$ for $i=2$ ($S^{2}$ is a unit sphere), by an infinite sum of SO(3)-irreducible subspaces yields the eigenfunctions
$\theta_{\kappa,l}$ of $H_{\kappa,i}^{0}$ which are of the form $C_{\kappa}r^{-1}u_{\kappa,l}(r)Y_{l\pi}(\Omega)$:
$C_{\kappa}$ the normalization constant, $u_{\kappa,l}$ as in eq.~(\ref{eq:uk}), $Y_{l\pi}$ the spherical harmonics normalized
to $1$ (recall that $\pi=\pi_{1}$ for $i=1$ and $\pi=\pi_{2}$ for $i=2$; the same for $l$ and the spherical angles $\Omega$). In
eq.~(\ref{eq:uk}), the parameters $A_{\kappa}=Z_{12}(\wp;\kappa,k)/\alpha$ and $B_{\kappa}=E_{\kappa}^{0}/\alpha$ for $i=1$, and
$A_{\kappa}=Z_{23}(\wp;\kappa,k)/\beta$ and $B_{\kappa}=E_{\kappa}^{0}/\beta$ for $i=2$. Here $\alpha$ and $\beta$ (as well as $\gamma$
in eq.~(\ref{eq:Hkappa})) are as in Corollary~\ref{cor:cor1}, and $Z_{12}(\wp;\kappa,k)$ and $Z_{23}(\wp;\kappa,k)$ are as in
Corollary~\ref{cor:cor4}. Therefore, if one solves eq.~(\ref{eq:uk}) with respect to $B_{\kappa}$ for both $i=1$ and $2$, then
the eigenvalues $E_{\kappa}^{0}$ are found from eq.~(\ref{eq:spec0}): $E_{\kappa}^{0}\propto B_{\kappa}$, where the coefficient of
proportionality is either $\alpha$ ($i=1$) or $\beta$ ($i=2$). Since $B_{\kappa}$ depends on $A_{\kappa}$ and $A_{\kappa}$ is a 
function of $\wp$, eq.~(\ref{eq:spec0}) allows one to establish $\wp$ as well.

Below we shall calculate $\disc(H_{\kappa}^{0})$ and in particular $\inf\disc(H_{\kappa})$ for integers $\kappa=0$ and $1$.

\subsubsection{\label{sec:kappa0}Bound states for $\kappa=0$.} 

Allowing $A_{0}<0$, solutions $u_{0,l}(r)$ to eq.~(\ref{eq:uk}) appear as a linear 
combination of the Whittaker \cite{Whi03} function $W_{n,l+1/2}(2r\sqrt{-B_{0}})$ and its linearly independent, in general, companion 
solution $M_{n,l+1/2}(2r\sqrt{-B_{0}})$, with $B_{0}=-A_{0}^{2}/(4n^{2})$ and $n=l+1,l+2,\ldots$ For $n>l$,
$W_{n,l+1/2}(z)=\bigl[(-1)^{n+l+1}(n+l)!/(2l+1)!\bigr]M_{n,l+1/2}(z)$ and thus $M_{n,l+1/2}$ and $W_{n,l+1/2}$ are linearly dependent. 
It suffices therefore to select one of them, say $M_{n,l+1/2}(z)$. The boundary conditions for $u_{0,l}(r)$ in $L^{2}(0,\infty;dr)$ as 
well as for $u_{0,l}(r)/r$ in $L^{2}(0,\infty;r^{2}dr)$ are fulfilled: $M_{n,l+1/2}(z)\to0$ and $M_{n,l+1/2}(z)/z\to\delta_{l0}$ as 
$z\to0$, and $M_{n,l+1/2}(z)\to0$ and $M_{n,l+1/2}(z)/z\to0$ as $z\to\infty$.

On the other hand, if $A_{0}/r$ is a repulsive potential, $A_{0}>0$, then $u_{0,l}(r)$ is represented by a linear
combination of functions $W_{-n,l+1/2}(2r\sqrt{-B_{0}})$ and $M_{-n,l+1/2}(2r\sqrt{-B_{0}})$. But
$M_{-n,l+1/2}(z)\to\infty$ as $z\to\infty$ and $W_{-n,l+1/2}(z)/z\to\infty$ as $z\to0$. Hence none of bound states are observed.
However, as pointed out by Albeverio et al. \cite[Theorem~2.1.3]{Alb04} (see also \cite{Est12}), a single bound state exists even if
$A_{0}\geq0$, provided that the Hamiltonian operator with $l=0$ is defined on a domain of one-parameter self-adjoint extensions.

For an attractive potential, eq.~(\ref{eq:spec0}) yields

\begin{equation}
\wp=\frac{n_{1}}{n_{2}}\sqrt{\frac{\alpha}{\beta}}
\label{eq:wppp}
\end{equation}

\noindent{}(refer to Lemma~\ref{lem:criter} for the definition of $\wp$) with integers $n_{1}=l_{1}+1,l_{1}+2,\ldots$ and 
$n_{2}=l_{2}+1,l_{2}+2,\ldots$ Equation~(\ref{eq:wppp}) indicates that eigenvalues $E_{0}^{0}$ are labeled by integers 
$n_{1},n_{2}=1,2,\ldots$ and $k=2,3,\ldots$, namely, $E_{0}^{0}=E(n_{1},n_{2},k)$, and

\begin{align}
\disc(H_{0}^{0})=&\inf_{D_{k}(\wp)\neq\varnothing}\Biggr\{
-\frac{1}{4}\biggl(\frac{Z_{12}}{n_{1}\sqrt{\alpha}}+\frac{Z_{3}}{n_{2}\sqrt{\beta}}
\Bigl(Z_{2}+\frac{Z_{1}}{c_{k}}\Bigr)\biggr)^{2}\co \nonumber \\
&c_{k}=\biggl(-\frac{Z_{13}}{Z_{12}+Z_{23}\wp^{k}}\biggr)^{1/k}\wp;n_{i}=l_{i}+1,l_{i}+2,\ldots; \nonumber \\
& l_{i}=0,1,\ldots; i=1,2\Biggr\}.
\label{eq:LLL1}
\end{align}

\noindent{}The procedure to find appropriate $k=p+1$ is described in \S\ref{sec:stability} (see Remarks~\ref{rem:13}--\ref{rem:15}). 
In particular, if $D_{\infty}(\wp)$ is nonempty (Corollary~\ref{cor:cor3}), that is, for $p=\infty$,
eq.~(\ref{eq:zeroV}), it holds $E(1,1,\infty)\leq E(n_{1},n_{2},\infty)\leq E(n_{1},n_{2},k)$, and eq.~(\ref{eq:LLL1}) is simplified to

\begin{subequations}\label{eq:inft}
\begin{align}
\disc(H_{0}^{0})=&\Biggl\{
-\frac{1}{4}\biggl(\frac{Z_{1}(Z_{2}+Z_{3})}{n_{1}\sqrt{\alpha}}+\frac{Z_{23}}{n_{2}\sqrt{\beta}}\biggr)^{2}\co 
0<\wp\leq1, \nonumber \\
& n_{i}=l_{i}+1,l_{i}+2,\ldots;l_{i}=0,1,\ldots;i=1,2 \Biggr\},
\label{eq:inft1} \\
=&\Biggl\{
-\frac{1}{4}\biggl(\frac{Z_{12}}{n_{1}\sqrt{\alpha}}+\frac{Z_{3}(Z_{1}+Z_{2})}{n_{2}\sqrt{\beta}}\biggr)^{2}\co 
\wp>1, \nonumber \\
& n_{i}=l_{i}+1,l_{i}+2,\ldots;l_{i}=0,1,\ldots;i=1,2 \Biggr\}.
\label{eq:inft2}
\end{align}
\end{subequations}

\begin{rem}\label{rem:19}
Assume that given $Z_{1}=Z_{2}=+1$, $Z_{3}=-1$ and $m_{1}\leq m_{3}$. Then $0<\wp\leq n_{1}/n_{2}$ and $k=2,3,\ldots$ for all 
$n_{1}\leq n_{2}$. By eq.~(\ref{eq:inft1}), the lower bound of $\disc(H_{0}^{0})$ equals $E(1,1,\infty)=-(4\beta)^{-1}$, which is, under 
the same conditions, in exact agreement with that given by Martin et al. \cite[eq.~(13)]{Mar92} (see also \cite[\S II~A]{Che90}, 
\cite[eq.~(3a)]{Fro92}). That is to say, for the bound three-unit-charge system, $\inf\disc(H_{0}^{0})$ is the lowest bound state energy 
at threshold. 
\end{rem}

\subsubsection{Bound states for $\kappa=1$.} 

We deduce from Corollary~\ref{cor:cor4} that $A_{1}>0$ in case $A_{0}<0$. Following the method developed by Nicholson \cite{Nic62}, for a 
repulsive potential $A_{1}/r^{2}$, we specify the eigenstates which result when this potential is cut off by an infinite repulsive core 
at $r=0$. Namely, the only solutions $u_{1,l}(r)$ in $L^{2}(0,\infty;dr)$ which vanish at $r=\infty$ are $\sqrt{r}K_{\nu}(r\sqrt{-
B_{1}})$, $\nu^{2}=A_{1}+(l+1/2)^{2}$, where $K_{\nu}(z)$ denotes the modified Bessel function of the second kind [we specify positive 
values of $\nu$ due to $K_{\nu}(z)=K_{-\nu}(z)$]. On the other hand, $u_{1,l}(r)$ is infinite at $r=0$. As demonstrated in \cite{Nic62}, 
the solutions $\sqrt{r}K_{\nu}(r\sqrt{-B_{1}})$ exist if $r>r_{0}$, provided that variables $\nu$ and $B_{1}$ satisfy 
$K_{\nu}(r_{0}\sqrt{-B_{1}})=0$; $r_{0}$ is known as the cut-off radius. This agrees with Case \cite{Cas50} who was the first to 
establish that for the potentials as singular as $r^{-2}$ or greater, bound states are determined up to the phase factor associated with 
$r_{0}$.

The solutions to $u_{1,l}(r_{0})=0$ are found by expanding $K_{\nu}(z)$ in terms of $I_{\nu}(z)$, the modified Bessel function
of the first kind. The result is $I_{\nu}(r_{0}\sqrt{-B_{1}})=I_{-\nu}(r_{0}\sqrt{-B_{1}})$ or equivalently,

$$
\sum_{n=1}^{\infty}\frac{ [(r_{0}/2)\sqrt{-B_{1}}]^{2n-2+\nu }}{(n-1)!\Gamma(n+\nu)}=
\sum_{n=1}^{\infty}\frac{ [(r_{0}/2)\sqrt{-B_{1}}]^{2n-2-\nu }}{(n-1)!\Gamma(n-\nu)}.
$$

\noindent{}Explicitly,

\begin{equation}
B_{1}=-\biggl(\frac{2}{r_{0}}\biggr)^{2}\biggl(\frac{\Gamma(n+\nu)}{\Gamma(n-\nu)}\biggr)^{1/\nu}\quad
\text{($\Gamma$ the Gamma function)}
\label{eq:B1}
\end{equation}

\noindent{}for integers $n>\nu$, and $B_{1}=0$ (no bound states) for integers $1\leq n\leq \nu$. Substitute $\nu^{2}=A_{1}+(l+1/2)^{2}$ 
in $n>\nu$ and, as in the case when $\kappa=0$, come by $n=l+1,l+2,\ldots$ Subsequently, the coupling constant is bounded by
$0<A_{1}<n^{2}-(l+1/2)^{2}\leq n^{2}-1/4$. It appears that $B_{1}<0$ is unbounded from below while the upper bound comes 
through $\nu\to n$. 

As in the case for $\kappa=0$, the eigenvalues $E_{1}^{0}\in\disc(H_{1}^{0})$ are found from eq.~(\ref{eq:spec0}), or in particular, 
from eq.~(\ref{eq:B1}). The result reads

\begin{align}
\disc(H_{1}^{0})=&\inf_{D_{k}(\wp)\neq\varnothing}
\Biggl\{E_{1}^{0}=\varkappa B_{1}\co \varkappa=\varkappa(i)=\biggl\{\begin{array}{ll}
\alpha, & i=1 \\ \beta, & i=2 \end{array}; B_{1}=B_{1}(i) \nonumber \\
&=-\biggl(\frac{2}{r_{0}}\biggr)^{2}\biggl(\frac{\Gamma(n_{i}+\nu_{i})}{\Gamma(n_{i}-\nu_{i})}\biggr)^{1/\nu_{i}};
\nu_{i}^{2}=A_{1}^{2}+(l_{i}+1/2)^{2}; 0<\nu_{i}<n_{i}; \nonumber \\
&n_{i}=l_{i}+1,l_{i}+2,\ldots; l_{i}=0,1,\ldots; r_{0}>0; \nonumber \\
&\alpha B_{1}(1)=\beta B_{1}(2)\Biggr\}.
\label{eq:E1}
\end{align}

\begin{rem}\label{rem:21}
The fact that $B_{1}$, eq.~(\ref{eq:B1}), is not bounded from below (for $n\to\infty$) does not necessary mean that $\disc(H_{1}^{0})$ 
is unbounded either, as this is still to be verified by solving $\alpha B_{1}(1)=\beta B_{1}(2)$ as in eq.~(\ref{eq:E1}). In agreement 
with Lemma~\ref{lem:criter}, it is apparent that solutions $\wp$ and $k\geq3$ exist only for appropriate integers $n_{1}$ and $n_{2}$ 
whose range strictly depends on masses (or equivalently, on multipliers $\alpha$, $\beta$). In those cases when none of common solutions 
are obtained, one should deduce that $E_{1}^{0}$ does not affect the total energy $E$ and higher eigenstates $E_{\kappa}$ 
($\kappa\geq2$), if such exist at all, should be added up to the series of $E$ for obtaining more accurate energies (see also 
Remark~\ref{rem:17} for analogous discussion in the case when $\kappa=0$). The numerical confirmation to it will be given below.
\end{rem}

\subsubsection{\label{sec:numerical}Some numerical results}

To illustrate the application of the approach presented in this paper, let us consider the helium atom (He) and the positronium 
negative ion (Ps$^{-}$). Although numerical methods to calculate bound states of these physical
systems are known in great detail for the most part due to Hylleraas \cite{Hyl29}, our goal is to comment on results following the 
analytic solutions obtained in the paper. The reason for choosing these atomic systems is due to different characteristics 
of the particles they are composed of. We shall calculate some lower bound states associated with the scalar SO(3) representation
($l=0$). On that account, the issue of possible function antisymmetrization is left out from further consideration as well as
Proposition~\ref{prop:1} holds. All calculations are performed in atomic units.

The helium atom contains two electrons ($Z_{1}=Z_{2}=-1$, $m_{1}=m_{2}=1$) and a nucleus ($Z_{3}=+2$, $m_{3}=7294.299536$). Here and
elsewhere below, $n_{1}=n_{2}=1$ ($l_{1}=l_{2}=0$). First, consider the case $\kappa=0$. Our task is to find integers $k\geq2$ such 
that $D_{k}(\wp)\neq\varnothing$. By eq.~(\ref{eq:wppp}), $\wp=0.707155<1$, thus $c_{k}=\wp(1/2-\wp^{k})^{-1/k}$ exists for all
$k\geq3$. The variables $\wp$, $c_{k}$ satisfy all necessary conditions in Corollary~\ref{cor:cor3} for all $k\geq3$. Subsequently,
$\inf\disc(H_{0})$ equals $E(1,1,\infty)=-2.914048$, by eq.~(\ref{eq:inft1}). Note that the result is invariant
under the change of charges and corresponding masses (see Remark~\ref{rem:16}) with $Z_{1}=+2$, $Z_{2}=Z_{3}=-1$ and 
$m_{1}=7294.299536$, $m_{2}=m_{3}=1$. In this case, $\wp=1.41412>1$ and $c_{k}=\wp(\wp^{k}/2-1)^{-1/k}$ exists for all $k\geq3$.
Again, the conditions in Corollary~\ref{cor:cor3} are fulfilled, and the minimal eigenvalue (refer to eq.~(\ref{eq:inft2})) is that 
obtained just above. In comparison, assume that $Z_{1}=Z_{3}=-1$, $Z_{2}=+2$. Then $\wp=1$, $k\geq2$ and the lower bound 
$E(1,1,\infty)=-(9/4)\alpha^{-1}=-4.499383$. As seen, the present eigenvalue is much lower than the above given one. A somewhat 
identical tendency is observed in all helium-like ions (Li$^{+}$, Be$^{2+}$, B$^{3+}$ etc.): while the two out of three arrangements 
provide the same eigenvalues, the third one differs and it is much lower than the other two. A distinctive feature of this particular 
arrangement is that the multiplier $\wp=1$, and vectors $\{\vecr_{ij}\}_{1\leq i<j\leq3}$ form the equilateral triangle at $k=\infty$, 
$\omega_{\infty}=\sigma_{\infty}=\pi-\tau_{\infty}=\pi/3$. To explain the appearance of solutions $-(9/4)\alpha^{-1}\simeq-9/2$, we
refer to the zeroth order perturbation theory which gives the energy $-Z_{2}^{2}=-4$, provided that the interaction potential 
$r_{13}^{-1}$ between two electrons $Z_{1}=-1$ and $Z_{3}=-1$ is neglected. This is not the case, as demonstrated above, for the 
remaining two arrangements because both interactions $r_{12}^{-1}$ and $r_{23}^{-1}$ are included in $H_{\kappa,i}^{0}$ explicitly;
see eq.~(\ref{eq:Hkappa1}).

\begin{figure}[htp!]
\centering
\subfigure[The contribution of $E_{1}$ does not affect the ground state energy of the helium 
atom, which is $E_{0}=-2.914048$ a.u. The curves are plotted for $k=3$.
]{\includegraphics[width=0.45\textwidth]{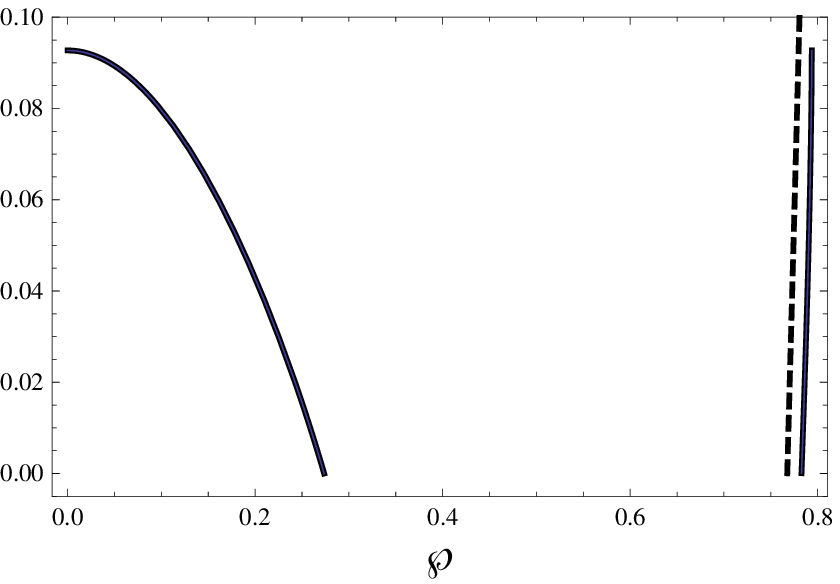}\label{fig:FigHe}}\quad\quad
\subfigure[The contribution of $E_{1}$ to the ground state energy of the positronium negative ion is well-defined with the inclusion of
the cut-off radius.]{\includegraphics[width=0.45\textwidth]{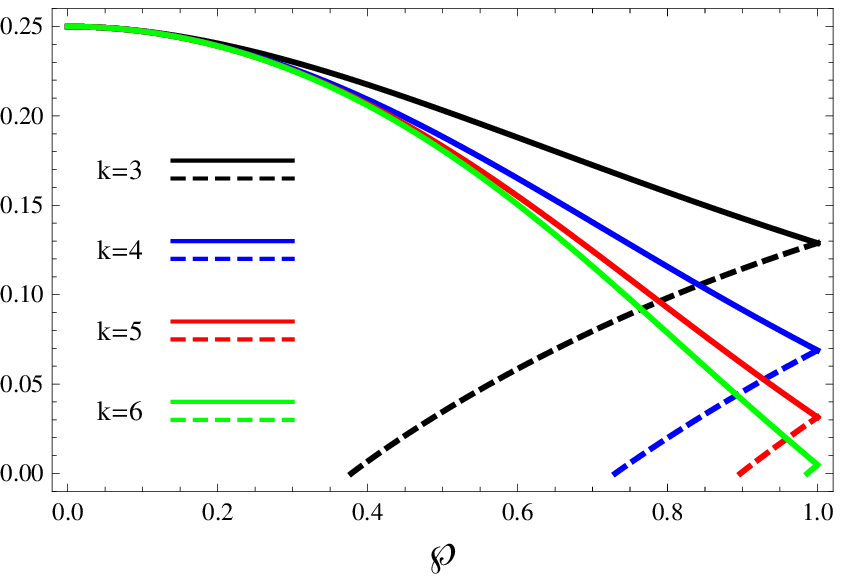}\label{fig:FigPs}}
\caption{\label{fig:E1}{\fontsize{11}{1}\selectfont (Color online) 
Solution to $\alpha B_{1}(1)=\beta B_{1}(2)$ with respect to $k$ and $\wp$; see eq.~(\ref{eq:E1}).
Solutions are found at the points, where the plane curves intersect with the 
same color (the same $k$) dashed curves.} }
\end{figure}

Second, consider the case $\kappa=1$. For the arrangement $Z_{1}=Z_{2}=-1$ and $Z_{3}=+2$, the coefficient 
$c_{k}=\wp(1/2-\wp^{k})^{-1/k}$, hence $1/2-\wp^{k}>0$ for all $k\geq3$. In this case, the estimate $1/2<\nu_{1},\nu_{2}<1$ yields $k=3$. 
However, none of $0<\wp<2^{-1/3}$ satisfy $\alpha B_{1}(1)=\beta B_{1}(2)$ in eq.~(\ref{eq:E1}), as it is clear from 
Fig.~\ref{fig:FigHe}. None of common solutions are obtained for the remaining two arrangements as well. Following \cite{Cas50}, 
therefore, we deduce that for the helium atom, the lowest (ground) eigenvalue of $H_{0}$ is obtained by the first expansion term 
$E_{0}^{0}$ in eq.~(\ref{eq:E}).

The positronium negative ion contains two electrons ($Z_{1}=Z_{3}=-1$ and $m_{1}=m_{3}=1$) and positron ($Z_{2}=+1$, $m_{2}=1$);
for $\kappa=0$, the other two arrangements are improper due to Lemma~\ref{lem:criter}: $\wp=(-Z_{1}/Z_{3})^{1/k}=1$ $\forall k\geq2$.
The lowest state $\inf\disc(H_{0})$ is found for $k=\infty$ ($\wp=1$, $\omega_{\infty}=\pi/3$), and it equals $E(1,1,\infty)=-1/4$. 
For $\kappa=1$, the bound $1/2<\nu_{1},\nu_{2}<1$ yields $\wp=1$, $k\geq3$. Then $\nu_{1}=\nu_{2}=(9/4-4^{1/k})^{1/2}$, hence $k=3,4,5,6$ 
(see Fig.~\ref{fig:FigPs}). By eq.~(\ref{eq:E1}),
$\inf\disc(H_{1})$ is at $k=3$ and it equals $E_{1}\simeq-0.515488/r_{0}^{2}$. On condition that the 
ground state of Ps$^{-}$ is $-0.261995$, by \cite{Mar92}, we find that the cut-off radius is $r_{0}\simeq 6.56$ (eg $r_{0}=4$ in 
\cite{Mil01}). Therefore, for the positronium negative ion, only the cut-off radius $r_{0}$ is needed to calculate the ground state 
energy.\bigskip{}


{\fontsize{10}{1}\selectfont \noindent{}\textbf{Acknowledgement}
The author is very grateful to Dr.~G.~Merkelis for very instructive and stimulating discussions. It is a pleasure to
thank Prof.~R.~Karazija for comments on an earlier version of the paper and Dr.~A.~Bernotas for attentive revision of the
present manuscript and for valuable remarks. }

\bibliographystyle{alpha}

\newcommand{\etalchar}[1]{$^{#1}$}
\providecommand{\noopsort}[1]{}\providecommand{\singleletter}[1]{#1}%

\end{document}